\documentclass[journal]{IEEEtran}
\usepackage{graphicx,color}
\usepackage{cite}
\usepackage{setspace} 
\usepackage{amsmath}
\usepackage{amsmath,amsthm,amssymb}
\usepackage{multirow}
\usepackage{hhline}
\usepackage{epsfig}
\usepackage{subfigure}
\usepackage{epstopdf}
\usepackage{verbatim}
\usepackage{algorithm}
\usepackage{algorithmicx}
\usepackage{algpseudocode}
\usepackage{etoolbox}
\usepackage{cases}
\usepackage{csquotes}
\usepackage{mathtools,stmaryrd}
\SetSymbolFont{stmry}{bold}{U}{stmry}{m}{n}
\usepackage{bbm}
\usepackage{lettrine}

\newcommand{\suchthat}{\;\ifnum\currentgrouptype=16 \middle\fi|\;}

\makeatletter
\newcommand*{\indep}{%
  \mathbin{%
    \mathpalette{\@indep}{}%
  }%
}
\newcommand*{\nindep}{%
  \mathbin{
    \mathpalette{\@indep}{\not}
  }%
}
\newcommand*{\@indep}[2]{%
  \sbox0{$#1\perp\m@th$}
  \sbox2{$#1=$}
  \sbox4{$#1\vcenter{}$}
  \rlap{\copy0}
  \dimen@=\dimexpr\ht2-\ht4-.2pt\relax
  \kern\dimen@
  {#2}%
  \kern\dimen@
  \copy0 
} 
\makeatother

\makeatletter
\newcommand*{\algrule}[1][\algorithmicindent]{%
  \makebox[#1][l]{%
    \hspace*{.2em}
    \vrule height .75\baselineskip depth .25\baselineskip
  }
}

\newcount\ALG@printindent@tempcnta
\def\ALG@printindent{%
    \ifnum \theALG@nested>0
    \ifx\ALG@text\ALG@x@notext
    \else
    \unskip
    \ALG@printindent@tempcnta=1
    \loop
    \algrule[\csname ALG@ind@\the\ALG@printindent@tempcnta\endcsname]%
    \advance \ALG@printindent@tempcnta 1
    \ifnum \ALG@printindent@tempcnta<\numexpr\theALG@nested+1\relax
    \repeat
    \fi
    \fi
}
\patchcmd{\ALG@doentity}{\noindent\hskip\ALG@tlm}{\ALG@printindent}{}{\errmessage{failed to patch}}
\patchcmd{\ALG@doentity}{\item[]\nointerlineskip}{}{}{} 
\makeatother

\newtheorem{theorem}{Theorem}

\newtheorem{corr}{Corollary}

\allowdisplaybreaks

\begin{document}

\title{Differential Chaos Shift Keying-based Wireless Power Transfer over a Frequency Selective Channel}

\author{\IEEEauthorblockN{Priyadarshi Mukherjee, Constantinos Psomas, and~Ioannis Krikidis}

\IEEEauthorblockA{IRIDA Research Centre for Communication Technologies\\
Department of Electrical and Computer Engineering, University of Cyprus\\
Emails: \{mukherjee.priyadarshi, psomas, krikidis\}@ucy.ac.cy}}

\maketitle
\begin{abstract}
This paper studies the performance of a differential chaos shift keying (DCSK)-based  wireless power transfer (WPT) setup in a frequency selective scenario. Particularly, by taking into account the nonlinearities of the energy harvesting (EH) process and a generalized frequency selective Nakagami-$m$ fading channel, we derive closed-form analytical expressions for the harvested energy in terms of the transmitted waveform and channel parameters. A simplified closed-form expression for the harvested energy is also obtained for a scenario, where the delay spread is negligible in comparison to the transmit symbol duration. Nontrivial design insights are provided, where it is shown how the power delay profile of the channel as well as the parameters of the transmitted waveform affect the EH performance. Our results show that a frequency selective channel is comparatively more beneficial for WPT compared to a flat fading scenario. However, a significant delay spread negatively impacts the energy transfer.
\end{abstract}


\IEEEpeerreviewmaketitle

\section{Introduction}

With the rapid evolution of the Internet of Things (IoT) in the recent years, the wireless traffic is expected to significantly increase between $2019$ and $2025$ \cite{ericsson}. For applications, where a large number of devices are deployed, the network lifetime is significantly affected due to their limited battery life. As a result, charging these devices becomes critical and hence, low-powered wireless communication networks is a relevant and important topic of research. In this context, based on the advances made in recent years, wireless power transfer (WPT) is emerging as a promising solution. WPT is especially useful for applications, where the devices are wirelessly powered by harvesting energy from ambient/dedicated radio-frequency (RF) signals \cite{wdesg}. This is achieved by employing a rectifying antenna (rectenna) at the receiver that converts the received RF signals to direct current (DC).

Accurate mathematical modeling of the energy harvesting (EH) circuit is extremely crucial for designing efficient WPT architectures. In this context, some works propose simplified models of the EH circuit, namely, linear \cite{lin}, piece-wise linear \cite{plinear}, and the nonlinear logistic saturation-based \cite{satm}. However, unlike these works, the authors in \cite{wdesg} propose a realistic circuit-based model of the harvester circuit that also enables the design of waveforms that maximize the WPT efficiency. Based on this model, the aspect of designing waveforms that result in an enhanced harvested energy gained importance. The authors in \cite{papr} show that the nonlinearity of the rectification process at the EH circuit causes certain waveforms, with high peak-to-average-power-ratio (PAPR) to provide higher output DC power, compared to conventional constant-envelop sinusoidal signals. Based on this observation, there are some works, which investigate the effect of the transmitted symbols and modulation techniques on WPT. By considering the nonlinear EH model proposed in \cite{wdesg}, the authors investigate the use of multisine waveforms for WPT due to their high PAPR. The work in \cite{hparam} proposes a simultaneous wireless information and power transfer (SWIPT) architecture based on the superposition of multi-carrier unmodulated and modulated waveforms at the transmitter. Apart from the multisine waveforms, experimental studies demonstrate that due to their high PAPR, chaotic waveforms outperform conventional single-tone signals in terms of the WPT efficiency \cite{chaosexp2}.

Due to its properties such as sensitivity to initial data and aperiodicity, chaotic waveforms have been extensively used in the past to improve the performance of wireless communication systems. In this context, the non-coherent modulation technique of differential chaos shift keying (DCSK) is one of the most widely studied chaotic signal-based communication system \cite{ch2}. The majority of the related works focus on the error performance of such systems for various scenarios. To exploit the benefits of both DCSK and WPT, there are few works in the literature that investigate DCSK-based SWIPT, e.g. \cite{chaoswipt1,chaoswipt4,chaoswipt3,jstsp}. The work in \cite{chaoswipt1} proposes a short-reference DCSK-based SWIPT architecture to achieve higher data rate than the conventional system. The authors in \cite{chaoswipt4} investigate adaptive link selection for buffer-aided relaying in a DCSK-SWIPT architecture, where two link-selection schemes based on harvested energy are proposed. The work in \cite{chaoswipt3} investigates a chaotic multi-carrier system in a SWIPT framework via the sub-carrier index to reduce the energy consumption. However, the aforementioned studies consider a simplified linear model for the EH, and as a result, they are independent of the circuit characteristics as well as the design of excitation waveforms  \cite{wdesg}. In this context, the authors in \cite{jstsp} propose a novel DCSK-based WPT architecture based on the nonlinearities of the EH process. Furthermore, they propose a WPT-optimal DCSK-based waveform that results in enhanced energy transfer. The harvested energy is analytically characterized in terms of the transmitted waveform parameters, by considering an ideal flat fading scenario. However, it is worth noting that DCSK-based signals are essentially wideband signals \cite{2ray}.

Motivated by this, in this paper, we investigate a point-to-point DCSK-based WPT topology by considering a frequency selective fading scenario. Closed-form analytical expressions of the harvested DC are obtained in terms of the transmitted waveform parameters and the power delay profile of a generalized Nakagami-$m$ multipath channel. By considering practical channel conditions, an approximation of the harvested energy is also provided, when the delay spread is negligible compared to the transmitted symbol duration. Our results demonstrate that while a frequency selective channel results in enhanced WPT when compared against its flat fading counterpart, a significant delay spread is detrimental for energy transfer.

\section{A Chaotic Signal-based WPT System Architecture} \label{SM}

\subsection{System model}
\begin{figure}[!t]
\centering\includegraphics[width=\linewidth]{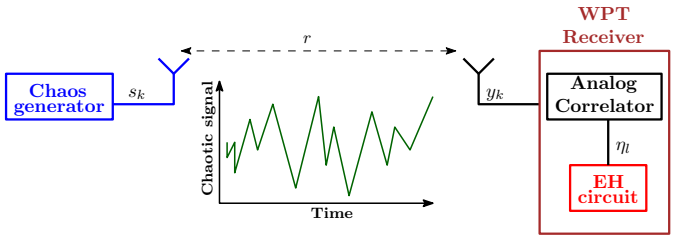}
\vspace{-1mm}
\caption{Architecture for DCSK-based WPT.}
\label{fig:model}
\vspace{-1mm}
\end{figure}

We consider a simple point-to-point WPT topology, where the transmitter employs a DCSK generator and the receiver consists of an analog correlator, followed by an EH circuit \cite{jstsp}, as shown in Fig. \ref{fig:model}. Note that, DCSK signals, until now, have been mainly considered for information transfer. However, here we focus on the WPT aspect and ignore the information side of the signal. We assume that the wireless link suffers from both large-scale path-loss effects and small-scale frequency-selective fading. Specifically, the received power is proportional to $r^{-a}$, where $r$ is the transmitter-receiver distance and $a>0$ denotes the pathloss exponent. Moreover, unlike \cite{jstsp}, a generalized  frequency selective fading channel is considered here. In wideband wireless communication systems, a commonly used channel model is the \textit{two-ray fading channel model} \cite{2ray}. Using this generalized channel model, the output of the channel is
\vspace{-1mm}
\begin{equation}  \label{chopt}
y_k = \alpha_1 s_k+\alpha_2s_{k-\tau},
\end{equation}
where $s_k$ is the $k$-th transmitted chaotic chip, $\tau$ is the time delay between two rays, and $\alpha_1,\alpha_2$ are independent Nakagami-$m$ distributed random variables. Generally, the two paths will have unequal average power gains and hence, we assume $\alpha_1,\alpha_2$ to have power gains $\Omega_1,\Omega_2$, such that $\Omega_1+\Omega_2=1$ \cite{2ray}.
In general, a frequency-selective channel adopts an $n$-ray $(n \geq 2)$ fading model. However, in this work, we consider a two-ray model for analytical simplicity. Furthermore, the aim of this work is to demonstrate how the chaotic signal-aided energy transfer is affected by the frequency-selective nature of the channel. Hence, the analysis can be extended to any arbitrary number of multipaths.

\subsection{Chaotic signals}

Assume a DCSK signal, where the current symbol is dependent on the previous symbol \cite{ch2} and different sets of chaotic sequences can be generated by using different initial conditions. Each transmitted bit is represented by two sets of chaotic signal samples, with the first set representing the reference, and the other conveying information. If $+1$ is to be transmitted, the data sample will be identical to the reference sample. Otherwise, an inverted version of the reference sample will be used as the data sample \cite{ch2}. Let $\beta$ be a non-negative integer, defined as the \textit{spreading factor}. Then, $2\beta$ chaotic samples are used to spread each information bit. During the $l$-th transmission interval, the output of the transmitter is
\begin{align}  \label{sym}
s_{l,k}=\begin{cases} 
x_{l,k}, & k=2(l-1)\beta+1,\dots,(2l-1)\beta,\\
d_lx_{l,k-\beta}, & k=(2l-1)\beta+1,\dots,2l\beta,
\end{cases}&
\end{align}
where $d_l=\pm1$ is the information bit, $x_{l,k}$ is the chaotic sequence used as the reference signal, and $x_{l,k-\beta}$ is its delayed version. Furthermore, $x_{l,k}$ can be generated according to various existing chaotic maps. Due to its good correlation properties, we consider the Chebyshev map of degree $\xi$ for chaotic signal generation, which is defined as \cite{ch2}
\begin{align}
x_{k+1}=\cos(\xi\cos^{-1}(x_k)), \forall \:\:|x_k| \leq 1.
\end{align}

\subsection{Analog correlator}   \label{AC}

The proposed WPT architecture employs an analog correlator, followed by an EH rectifier circuit. An analog correlator essentially consists of a series of $(\psi-1)$ delay blocks, where $\psi$ is a positive integer; the rationale behind this application is that, the signal can be effectively integrated over a certain time interval \cite{anaco2}. In what follows, for the sake of simplicity, we will consider $\psi$ equal to the transmitted symbol length, i.e. $\psi=2\beta$. As a result, the correlator output $\eta_l$ for the $l$-th received symbol is
\begin{align}    \label{corr}
\eta_l&=\sqrt{P_t}\sum\limits_{k=2(l-1)\beta+1}^{2l\beta} y_{l,k} \nonumber \\
&=\sqrt{P_t}\sum\limits_{k=2(l-1)\beta+1}^{2l\beta} \alpha_1 s_k+\alpha_2s_{k-\tau},
\end{align}
which follows from \eqref{chopt} and $P_t$ is the transmission power. It is worthy to note that $\psi=1$ corresponds to the conventional case without a correlator\footnote{An analog correlator is placed at the receiver, prior to the harvester, because the work in \cite{jstsp} analytically proves and demonstrates why the correlator is beneficial for energy transfer \cite[Proposition $1$]{jstsp}.}.

\subsection{Energy transfer model} \label{ehsec}

The WPT receiver is equipped with an antenna followed by a rectifier. The rectifier, which generally consists of a diode  (e.g., a Schottky diode) and a passive low pass filter, acts as an envelope detector \cite{envd} and therefore neglects the phase of the received signal $y$. Based on the nonlinearity of this circuit, the harvested power in terms of $y$ is \cite{wdesg}
\begin{equation} \label{brunoeh}
P_{\rm harv}=k_2R_{ant}\mathbb{E} \{ |y|^2 \}+k_4R_{ant}^2\mathbb{E} \{ |y|^4 \},
\end{equation}
where the parameters $k_2,k_4,$ and $R_{ant}$ are constants determined by the characteristics of the circuit. The conventional linear model is a special case of this nonlinear model and can be obtained by considering only the first term in \eqref{brunoeh}. Moreover, the energy transfer model considered in this work is based on the assumption that the harvester operates in the nonlinear region \cite{wdesg}. If the power of the harvester input signal becomes too large, the diode inside the harvester will be forced into the saturation region of operation, making the derived analytical results inapplicable. For the sake of presentation, we will use $\varepsilon_{1} = r^{-a}k_2R_{ant}P_t$ and $\varepsilon_{2}=r^{-2a}k_4R_{ant}^2P_t^2$. Hence, by incorporating a $\psi$-bit analog correlator at the receiver prior to the harvester, the average harvested power is \cite{jstsp}
\begin{align} \label{zdef1}
P_{{\rm harv}}&=\varepsilon_1\mathbb{E}\left\lbrace \left( \sum\limits_{k=1}^{\psi}y_k \right) ^2 \right\rbrace  + \varepsilon_2\mathbb{E}\left\lbrace \left( \sum\limits_{k=1}^{\psi}y_k \right) ^4 \right\rbrace.
\end{align}

\section{Analysis of harvested energy}

In this section, we analyse the harvested energy in terms of the transmitted chaotic waveform parameters and also the frequency-selective channel parameters.

The authors in \cite{jstsp}, proposed a WPT-optimal signal design of length $\beta+1$, where during the $l$-th symbol duration, the transmitter output is
\begin{align}  \label{optsym}
&s_k \nonumber \\
&=\!\begin{cases} 
\!x_k, & \!\!k=(\beta+1)(l-1)+1,\\
\!d_lx_{(\beta+1)(l-1)+1}, & \!\!k=(\beta+1)(l-1)+2,\cdots,(\beta+1)l.
\end{cases}
\end{align}
If this waveform is transmitted by considering a Nakagami-$m$ block fading scenario and a $\beta+1$ bit analog correlator is implemented at the receiver prior to the harvester, the obtained harvested power is \cite{jstsp}
\begin{equation}  \label{ff}
P_{\rm harv}=\frac{1}{2}\varepsilon_1(1+\beta^2)+\frac{3(1+m)}{8m}\varepsilon_2(1+6\beta^2+\beta^4).
\end{equation}
However, this is an ideal scenario and in practice, chaotic signals are wideband signals \cite{2ray}. As a result, to obtain a realistic insight into the harvested power of this WPT-optimal waveform, we consider here a frequency selective scenario as defined in \eqref{chopt}, where we have the delayed copy of the $l$-th transmitted symbol, i.e. $s_{k-\tau}$ as
\begin{equation}    \label{delayed}
s_{k-\tau}=
\begin{cases}
  d_{l-1}x_{(\beta+1)(l-2)+1}, & 
       \begin{aligned}[t]
       k&=(\beta+1)(l-1)+1,\\
       &\quad \cdots,(\beta+1)(l-1)+\tau,
       \end{aligned}
\\
  x_{(\beta+1)(l-1)+1}, & k=(\beta+1)(l-1)+\tau+1,
\\
d_lx_{(\beta+1)(l-1)+1}, &
       \begin{aligned}[t]
       k&=(\beta+1)(l-1)+\tau+2,\\
       &\quad \cdots,(\beta+1)l.
       \end{aligned}
\end{cases}
\end{equation}
Hence, by using \eqref{chopt}, \eqref{corr}, \eqref{optsym}, and  \eqref{delayed} in \eqref{zdef1}, we obtain
\begin{align} \label{zdef}
P_{\rm harv}&=\varepsilon_1\mathbb{E}\left\lbrace \left(  \sum\limits_{k=(\beta+1)(l-1)+1}^{(1+\beta)l}y_k \right) ^2 \right\rbrace \nonumber \\
& \quad  + \varepsilon_2\mathbb{E}\left\lbrace \left(\sum\limits_{k=(\beta+1)(l-1)+1}^{(1+\beta)l}y_k \right) ^4 \right\rbrace,
\end{align}
where the expectation is taken over $\alpha_1,\alpha_2,$ and $s_k$. Towards this direction, we provide the following theorem.

\begin{theorem}    \label{theo1}
The harvested power for the WPT-optimal signal in a frequency selective Nakagami-$m$ fading scenario is given by \eqref{theq1}.
\begin{figure*}
\begin{align}    \label{theq1}
P_{\rm harv}&=\frac{\varepsilon_1}{2} \left( \tau^2\Omega_2 + \Omega_1 \left(1+\beta^2 \right)+\Omega_2 \left(1+(\beta-\tau)^2 \right) + 2 \left(1+\beta^2-\beta\tau\right) \left(\frac{\Gamma(m+0.5)}{\Gamma(m)}\right)^2 \frac{\sqrt{\Omega_1\Omega_2}}{m} \right) \nonumber \\
& \quad  + \frac{\varepsilon_2}{2} \left( \frac{3}{4} \left( \tau^4\Omega_2^2 \left( \frac{m+1}{m} \right)+ \Omega_1^2 \left( 1+6\beta^2+\beta^4 \right) \left( \frac{m+1}{m} \right) +4\frac{\Gamma (1.5+m) \Gamma (0.5+m)}{m^2\Gamma^2(m)}\Omega_1^{1.5}\Omega_2^{0.5} \right. \right. \nonumber \\
& \quad\left.\left. \times \left( \left(1+3\beta^2\right)+\left(\beta-\tau\right)\left(3\beta+\beta^3 \right) \right)+ 6\Omega_1\Omega_2 \left( \left(\beta-\tau\right)^2 \left(1+\beta^2\right) +4\beta\left(\beta-\tau\right) +1+\beta^2 \right) \right. \right. \nonumber \\
& \quad \left. \left. + 4\frac{\Gamma
(1.5+m) \Gamma (0.5+m)}{m^2\Gamma(m)^2}\Omega_1^{0.5}\Omega_2^{1.5} \left( 1+3 \left(\beta-\tau\right)^2+\beta\left(\beta-\tau\right)^3 +3\beta\left(\beta-\tau\right) \right)+\Omega_2^2 \left( 1+6(\beta-\tau)^2 \right. \right.\right. \nonumber \\
& \quad \left.\left. \left. +(\beta-\tau)^4 \right) \left( \frac{m+1}{m}\right)  \right) + 3\tau^2\Omega_2 \left( \Omega_1 \left(1+\beta^2 \right)+\Omega_2 \left(1+(\beta-\tau)^2 \right)+2\left(1+\beta^2-\beta\tau\right) \left(\frac{\Gamma(m+0.5)}{\Gamma(m)}\right)^2 \right. \right. \nonumber \\
& \quad \left. \left. \frac{\sqrt{\Omega_1\Omega_2}}{m} \right) \right).
\end{align}
\hrulefill
\end{figure*}
\end{theorem}

\begin{proof}
See Appendix \ref{app1}. 
\end{proof}

We observe from \eqref{theq1} how the power distribution over the two paths, i.e. $\Omega_1,\Omega_2$ and the delay spread $\tau$ affect the WPT performance. The theorem demonstrates that a considerable $\tau$ is detrimental for energy transfer. Note that the flat fading channel is a special case of its frequency selective counterpart. By  replacing $\Omega_1=1,\Omega_2=0,$ and considering delay spread $\tau=0$ in \eqref{theq1}, $P_{\rm harv}$ reduces to \eqref{ff}. Furthermore, in a frequency selective scenario, for an $N$-tone multisine signal, the linear term of the harvested power is independent of $N$ and the nonlinear term is linearly dependent on $N$ \cite{wdesg}. On the contrary, with a DCSK-based waveform, the linear and nonlinear terms of $P_{\rm harv}$ in \eqref{theq1} are proportional to $\beta^2$ and $\beta^4$, respectively. As a result, this waveform significantly outperforms multisine waveforms, in terms of WPT.

To the best of our knowledge, this is the first analytical closed-form expression obtained for DCSK-based WPT, which  characterizes the harvested energy in terms of the power delay profile of a frequency selective wireless channel and also the nonlinearities of the EH process at the harvester. Finally, note that \eqref{theq1} is an extremely complex and lengthy equation. However, the work in \cite{2ray} demonstrates that practically we have $0<\tau \ll \beta$. As a result, an approximation for $P_{\rm harv}$ can be obtained corresponding to the $\tau \rightarrow 0$ scenario as an upper bound, given in the following corollary.

\begin{corr}
By considering the limiting case of negligible delay spread, the harvested power is given by
\begin{align}
\lim\limits_{\tau \rightarrow 0} P_{\rm harv}\!&=\!\frac{\varepsilon_1}{2}\left(1+\beta^2 \right)\!\left( \Omega_1+\Omega_2+2\left(\frac{\Gamma(m+0.5)}{\Gamma(m)}\right)^2 \right. \nonumber \\
&\quad \left. \times \frac{\sqrt{\Omega_1\Omega_2}}{m} \right)+\frac{3\varepsilon_2}{8} \left( 1+6\beta^2+\beta^4 \right) \nonumber \\
& \quad\times \left( \left( \Omega_1^2+\Omega_2^2 \right) \left( \frac{m+1}{m} \right)+6\Omega_1\Omega_2 \right. \nonumber \\
& \quad \left. 4\frac{\Gamma
(1.5+m) \Gamma (0.5+m)}{m^2 \Gamma(m)^2}\!\! \left( \Omega_1^{0.5}\Omega_2^{1.5}+\Omega_1^{1.5}\Omega_2^{0.5} \right)\!\!\right)\!\!.
\end{align}
\end{corr}
The above corollary follows directly from Theorem \ref{theo1} by replacing $\tau \rightarrow 0$ in \eqref{theq1}. Note that $m=1$ leads to an enhanced harvested energy compared to the $m \rightarrow \infty$ scenario, i.e. wireless fading enhances WPT. This observation corroborates the claims made in \cite{jstsp} regarding the beneficial role of fading in DCSK-based WPT systems.

\section{Numerical Results}

We validate our theoretical analysis through extensive Monte-Carlo simulations. Unless otherwise stated, lines correspond to analysis whereas markers correspond to simulation results. Without any loss of generality, we consider a transmission power of $P_t=30$ dBm, a Tx-Rx distance $r=20$ m, and a pathloss exponent $\alpha=4$. The parameters considered for the EH model are taken as $k_2=0.0034,k_4=0.3829,$ and $R_{ant}=50$ $\Omega$ \cite{hparam}.

\begin{figure}[!t]
 \centering\includegraphics[width=0.96\linewidth]{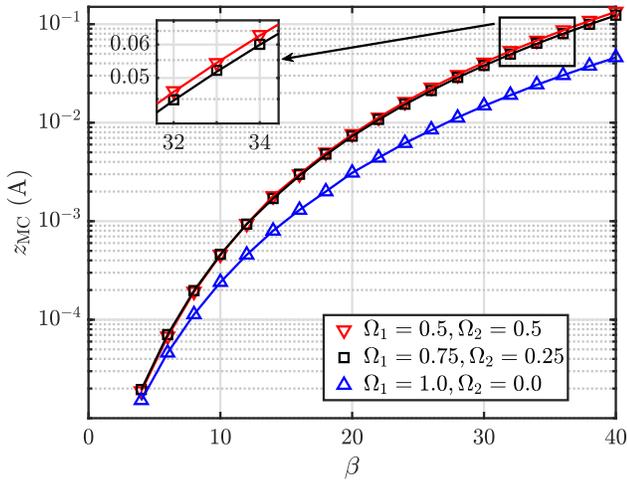}
 \vspace{-2mm}
\caption{WPT performance over a two-ray Nakagami-$m$ fading channel ($z_{\rm MC}$ versus $\beta$) with $m=4,\tau=3$.}
\label{fig:fig1}
\vspace{-1mm}
\end{figure}

\begin{figure}[!t]
\centering\includegraphics[width=0.96\linewidth]{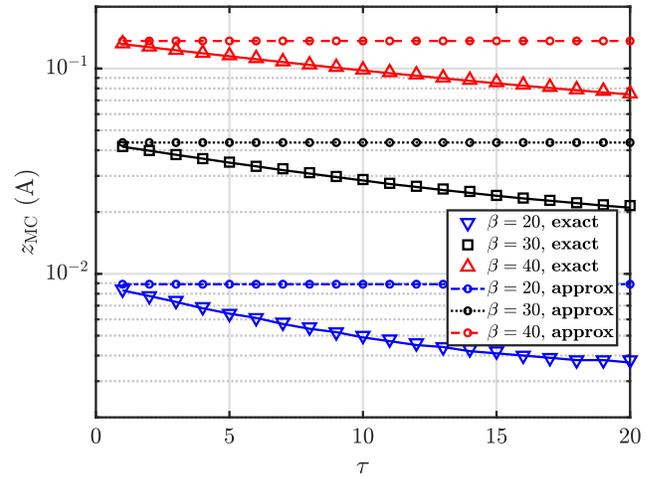}
 \vspace{-2mm}
\caption{WPT performance over a two-ray Nakagami-$m$ fading channel ($z_{\rm MC}$ versus $\tau$) with $m=4,\Omega_1=0.75,\Omega_2=0.25$.}
\label{fig:fig2}
\vspace{-1mm}
\end{figure}


Fig. \ref{fig:fig1} depicts the WPT performance of the proposed DCSK-based energy waveform with respect to the spreading factor, over a two ray Nakagami-$m$ frequency selective channel. Firstly, we observe the theoretical results (lines) match very closely with the simulation
results (markers); this verifies our proposed analytical framework. Moreover, the figure also demonstrates that the frequency selective nature of the channel, i.e. $\Omega_2\neq0$ is beneficial for WPT and the worst performance is obtained as the channel moves towards the limiting case of a flat fading scenario, i.e. $\Omega_1=1,\Omega_2=0.$ While \cite{jstsp} proved that fading is beneficial for DCSK-based WPT, we extend the claim in this work by demonstrating that frequency selective fading is more beneficial for DCSK-based WPT, when compared against its flat fading counterpart.

Fig. \ref{fig:fig2} illustrates the impact of the delay spread $\tau$ on the WPT performance of the proposed energy waveform for various values of the spreading factor $\beta$. The figure demonstrates that the harvested energy decreases with $\tau$, which is inline with \eqref{theq1} in Theorem \ref{theo1}. Moreover, for a particular $\tau$, greater $\beta$ results in higher harvested DC, which corroborates the claim made in Fig. \ref{fig:fig1}. Finally, we observe that the effect of $\tau$ can be neglected when $\tau \ll \beta$. However, for a considerable $\tau$, the approximation does not agree with the exact harvested energy. Fortunately, in most practical applications, the condition $\tau \ll \beta$ holds \cite{2ray} and as a result, the degradation in WPT performance due to delay spread of the channel is negligible. Finally, we observe that the gap between the actual and the approximated harvested energy decreases with increasing $\beta$; for example, observe the performance gap at $\tau=10$ between $\beta=20$, $\beta=30$, and $\beta=40$.

Fig. \ref{fig:fig3} demonstrates the combined effect of the frequency-selectivity and delay spread of the channel. In this figure we obtain the harvested DC by jointly varying the power ratio $\frac{\Omega_2}{\Omega_1}$ and delay spread $\tau$. It is worthy to note that in this figure, $\frac{\Omega_2}{\Omega_1} \in \left( 0,1 \right)$; while $\frac{\Omega_2}{\Omega_1}=0$ denotes a flat fading channel, $\frac{\Omega_2}{\Omega_1}=1$ implies equal power distribution between the two paths of the considered two-ray path model. We do not consider $\Omega_2>\Omega_1$ because for all practical scenarios, we always have $\Omega_2 \leq \Omega_1$ \cite{2ray}. An interesting observation obtained from Fig. \ref{fig:fig1} and Fig. \ref{fig:fig2} is that while frequency selectivity of the channel enhances WPT performance, increasing delay spread $\tau$ has a negative impact on the same. As $\tau \neq 0$ is the signature of a frequency selective fading scenario, the best performance is observed at $\frac{\Omega_2}{\Omega_1}=1$, i.e. $\Omega_1=\Omega_2=0.5$ and $\tau=1$.

\begin{figure}[!t]
\centering\includegraphics[width=0.96\linewidth]{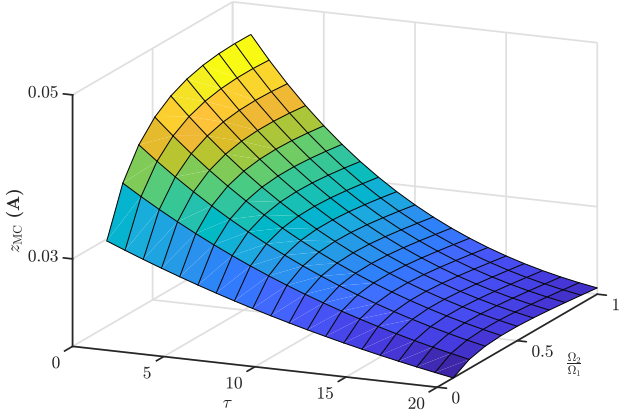}
 \vspace{-2mm}
\caption{Joint effect of the multipath components and the delay spread on the harvested energy with $m=4,\beta=30$.}
\label{fig:fig3}
\vspace{-1mm}
\end{figure}

\section{Conclusion}

In this paper, we investigated the framework of DCSK-based WPT by considering a generalized frequency selective Nakagami-$m$ fading scenario and also taking into account the nonlinearities of the EH process. To the best of our knowledge, this is the first work where a closed-form expression of harvested energy in a frequency-selective fading scenario is obtained as a function of the multipath components and the delay spread. An approximation of the harvested DC is also provided for the scenario, where the delay spread is negligible compared to the transmitted symbol duration. It was shown that while a frequency selective channel enhances WPT compared to its flat fading counterpart, an increasing delay spread is detrimental to wireless energy transfer. While the analytical expressions are derived by considering a simplified two-ray fading channel model, the analysis can be generalized for any arbitrary number of multipaths.

\appendices

\section{Proof of Theorem \ref{theo1}}  \label{app1}

From \eqref{chopt}, we have the channel output
\begin{equation}
y_k = \alpha_1 s_k+\alpha_2s_{k-\tau},
\end{equation}
where $\alpha_1,\alpha_2$ are independent Nakagami-$m$ distributed random variables with power gains $\Omega_1,\Omega_2$, respectively such that $\Omega_1+\Omega_2=1$. Therefore, the probability density function (PDF) of $\alpha_i,$ $\forall$ $i \in \{1,2\}$ is
\begin{equation}    \label{nakadef}
f_{\alpha_i}(z)=\frac{2m^mz^{2m-1}e^{-\frac{mz^2}{\Omega_i}}}{\Gamma(m)\Omega_i^m}, \:\: \forall \:\:z \geq 0,
\end{equation}
where $\Gamma(\cdot)$ denotes the complete Gamma function, $m \geq 1$ controls the severity of the amplitude fading, and $\mathbb{E}\{\alpha_i^2\}=\Omega_i$. Moreover, the chaotic sequences generated from the considered Chebyshev map have the following invariant PDF \cite{ch2}
\begin{align}  \label{spdf}
f_X(x)=\begin{cases} 
\frac{1}{\pi\sqrt{1-x^2}}, & |x|< 1,\\
0, & \text{otherwise}.
\end{cases}&
\end{align}
Here, we aim to obtain a closed-form expression for
\begin{align}    \label{happeq}
P_{\rm harv}&=\varepsilon_1\mathbb{E}\left\lbrace \left(  \sum\limits_{k=(\beta+1)(l-1)+1}^{(1+\beta)l}y_k \right) ^2 \right\rbrace \nonumber \\
& \quad  + \varepsilon_2\mathbb{E}\left\lbrace \left(\sum\limits_{k=(\beta+1)(l-1)+1}^{(1+\beta)l}y_k \right) ^4 \right\rbrace,
\end{align}
where the expectation is taken over $\alpha_1,\alpha_2,$ and $s_k$. In this context, from \eqref{optsym} and \eqref{delayed}, we obtain
\begin{align}
&\sum\limits_{k=(\beta+1)(l-1)+1}^{(1+\beta)l}y_k =\sum\limits_{k=(\beta+1)(l-1)+1}^{(1+\beta)l} \alpha_1 s_k+\alpha_2s_{k-\tau} \nonumber \\
&=\tau \alpha_2 d_{l-1} x_{(\beta+1)(l-2)+1} \nonumber \\
& \quad+ \left( \alpha_1 \left(1+\beta d_l \right)+\alpha_2 \left(1+(\beta-\tau) d_l \right) \right) x_{(\beta+1)(l-1)+1} \nonumber \\
&=\delta_1x_{(\beta+1)(l-2)+1}+\delta_2x_{(\beta+1)(l-1)+1},
\end{align}
where $\delta_1=\tau \alpha_2 d_{l-1}$ and $\delta_2=\alpha_1 (1+\beta d_l)+\alpha_2(1+(\beta-\tau) d_l)$. Accordingly, \eqref{happeq} is rewritten as
\begin{align}  \label{redef1}
P_{\rm harv}&=\varepsilon_1\mathbb{E}\left\lbrace \left( \delta_1x_{(\beta+1)(l-2)+1}+\delta_2x_{(\beta+1)(l-1)+1} \right) ^2 \right\rbrace \nonumber \\
& \quad+\varepsilon_2\mathbb{E}\left\lbrace \left( \delta_1x_{(\beta+1)(l-2)+1}+\delta_2x_{(\beta+1)(l-1)+1} \right) ^4 \right\rbrace.
\end{align}
Then, the first term of \eqref{redef1} can be evaluated as
\begin{align}    \label{pow2}
&\mathbb{E}\left\lbrace \left( \delta_1x_{(\beta+1)(l-2)+1}+\delta_2x_{(\beta+1)(l-1)+1} \right) ^2 \right\rbrace \nonumber \\
&=\mathbb{E} \left\lbrace \delta_1^2x_{(\beta+1)(l-2)+1}^2 + \delta_2^2x_{(\beta+1)(l-1)+1}^2 \right. \nonumber \\
& \quad \left. + 2\delta_1\delta_2x_{(\beta+1)(l-2)+1}x_{(\beta+1)(l-1)+1} \right\rbrace \nonumber \\
&\overset{(a)}{=}\frac{1}{2} \left( \mathbb{E} \left\lbrace \delta_1^2 \right\rbrace+\mathbb{E} \left\lbrace \delta_2^2 \right\rbrace \right) \nonumber\\
&\overset{(b)}{=} \frac{1}{2} \left( \tau^2\Omega_2 + \Omega_1 \left(1+\beta^2 \right)+\Omega_2 \left(1+(\beta-\tau)^2 \right) \right. \nonumber \\
& \quad\left. + 2 \left(1+\beta^2-\beta\tau\right) \left(\frac{\Gamma(m+0.5)}{\Gamma(m)}\right)^2\frac{\sqrt{\Omega_1\Omega_2}}{m} \right),
\end{align}
where $(a)$ follows from \eqref{spdf} as $\mathbb{E}\{ x_i\}=0$ and $\mathbb{E}\{ x_i^2\}=\frac{1}{2}$ $\forall$ $i$. Moreover, by assuming equally likely transmissions of $d_l,d_{l-1}=\pm 1$, $(b)$ follows from
\begin{align}  \label{d1d2sq}
\mathbb{E} \left\lbrace \delta_1^2 \right\rbrace&=\mathbb{E}\left\lbrace (\tau \alpha_2 d_{l-1})^2 \right\rbrace=\tau^2\mathbb{E}\left\lbrace \alpha_2^2 \right\rbrace\mathbb{E}\left\lbrace d_{l-1}^2 \right\rbrace=\tau^2\Omega_2, \nonumber \\
\mathbb{E} \left\lbrace \delta_2^2 \right\rbrace&=\mathbb{E}  \left\lbrace \left( \alpha_1 \left(1+\beta d_l \right)+\alpha_2 \left(1+(\beta-\tau) d_l \right) \right)^2 \right\rbrace \nonumber \\
&=\mathbb{E}  \left\lbrace \alpha_1^2 \left(1+\beta d_l \right)^2+\alpha_2^2 \left(1+(\beta-\tau) d_l \right)^2 \right. \nonumber \\
& \quad \left. +2 \alpha_1\alpha_2\left(1+\beta d_l \right)\left(1+(\beta-\tau) d_l \right) \right\rbrace \nonumber \\
&=\Omega_1 \left(1+\beta^2 \right)+\Omega_2 \left(1+(\beta-\tau)^2 \right) \nonumber \\
& \quad +2\left(1+\beta^2-\beta\tau\right) \left(\frac{\Gamma(m+0.5)}{\Gamma(m)}\right)^2\frac{\sqrt{\Omega_1\Omega_2}}{m},
\end{align}
where
\begin{align}  \label{dnl1}
\mathbb{E}\left\lbrace \! \left(  1+\beta d_l \right) ^2\! \right\rbrace&=\frac{1}{2}\left\lbrace \left(  1+\beta \right)^2+\left(1-\beta\right)^2 \right\rbrace =1+\beta^2, \nonumber \\
\mathbb{E}\left\lbrace \!\left(  1+(\beta-\tau) d_l \right) ^2\! \right\rbrace&=\!\frac{1}{2}\!\left\lbrace \! \left(  1+(\beta-\tau) \right)^2\!\!+\!\!\left(1-(\beta-\tau)\right)^2 \! \right\rbrace \nonumber \\
& =1+(\beta-\tau)^2.
\end{align}
Furthermore, we also use
\begin{align}    \label{val1}
\mathbb{E}\{\alpha_i\}&=\frac{2m^m}{\Gamma(m)\Omega_i^m}\int_0^{\infty} z^{2m}e^{-\frac{mz^2}{\Omega_i}} dz \nonumber \\
&=\frac{1}{\Gamma(m)} \sqrt{\frac{\Omega_i}{m}} \int_0^{\infty} v^{m-0.5}e^{-v}dv\!=\!\frac{\Gamma(m+0.5)}{\Gamma(m)}\sqrt{\frac{\Omega_i}{m}},
\end{align}
which follows from the transformation $ \frac{mz^2}{\Omega_i} \rightarrow v$.

Similarly, by using the multinomial
theorem, the second term of \eqref{redef1} can be expanded as
\begin{align}  \label{pow4}
&\mathbb{E}\left\lbrace \left( \delta_1x_{(\beta+1)(l-2)+1}+\delta_2x_{(\beta+1)(l-1)+1} \right) ^4 \right\rbrace \nonumber \\
&\overset{(a)}{=}\mathbb{E} \left\lbrace \delta_1^4x_{(\beta+1)(l-2)+1}^4 + 6\delta_1^2\delta_2^2x_{(\beta+1)(l-2)+1}^2x_{(\beta+1)(l-1)+1}^2 \right. \nonumber \\
& \quad \left.+ \delta_2^4x_{(\beta+1)(l-1)+1}^4 \right\rbrace \nonumber \\
&\overset{(b)}{=}\frac{3}{8}\mathbb{E} \left\lbrace \delta_1^4 \right\rbrace + \frac{3}{2}\mathbb{E} \left\lbrace \delta_1^2 \right\rbrace\mathbb{E} \left\lbrace \delta_2^2 \right\rbrace + \frac{3}{8} \mathbb{E} \left\lbrace \delta_2^4 \right\rbrace,
\end{align}
where $(a)$ follows from $\mathbb{E} \left\lbrace x_i^s \right\rbrace=0 $ $\forall$ $i$ when $s$ is odd and $(b)$ follows from $\mathbb{E} \left\lbrace x_i^2 \right\rbrace=\frac{1}{2} $ $\forall$ $i$. Furthermore, $\mathbb{E}\left\lbrace \delta_1^2 \right\rbrace,\mathbb{E}\left\lbrace \delta_2^2 \right\rbrace$ are obtained from \eqref{d1d2sq} and
\begin{align}
&\mathbb{E} \left\lbrace \delta_1^4 \right\rbrace=\mathbb{E}\left\lbrace (\tau \alpha_2 d_{l-1})^4 \right\rbrace \nonumber \\
&=\tau^4\mathbb{E}\left\lbrace \alpha_2^4 \right\rbrace\mathbb{E}\left\lbrace d_{l-1}^4 \right\rbrace=\tau^4\Omega_2^2 \left( \frac{m+1}{m} \right), \nonumber \\
&\mathbb{E} \left\lbrace \delta_2^4 \right\rbrace=\mathbb{E}  \left\lbrace \left( \alpha_1 \left(1+\beta d_l \right)+\alpha_2 \left(1+(\beta-\tau) d_l \right) \right)^4 \right\rbrace \nonumber \\
&=\mathbb{E}  \left\lbrace \alpha_1^4 \left(1+\beta d_l \right)^4+4\alpha_1^3\alpha_2\left(1+\beta d_l \right)^3\left(1+(\beta-\tau) d_l \right) \right. \nonumber \\
& \quad \left. + 6\alpha_1^2\alpha_2^2\left(1+\beta d_l \right)^2\left(1+(\beta-\tau) d_l \right)^2  \right. \nonumber \\
&\quad +4\alpha_1\alpha_2^3\left(1+\beta d_l \right)\left(1+(\beta-\tau) d_l \right)^3 \nonumber \\
& \quad + \left. \alpha_2^4 \left(1+(\beta-\tau) d_l \right)^4 \right\rbrace \nonumber \\
&=\Omega_1^2 \left(\! 1+\!6\beta^2+\!\beta^4 \right)\!\! \left(\! \frac{m+1}{m}\! \right)\! +\! 4\frac{\Gamma
(1.5+m) \Gamma (0.5+m)}{m^2\Gamma^2(m)} \nonumber \\
& \quad \times \Omega_1^{1.5}\Omega_2^{0.5}\left( \left(1+3\beta^2\right)+\left(\beta-\tau\right)\left(3\beta+\beta^3 \right) \right) \nonumber \\
& \quad + 6\Omega_1\Omega_2 \left( \left(\beta-\tau\right)^2 \left(1+\beta^2\right) +4\beta\left(\beta-\tau\right)+1+\beta^2 \right) \nonumber \\
& \quad + 4\frac{\Gamma
(1.5+m) \Gamma (0.5+m)}{m^2\Gamma(m)^2}\Omega_1^{0.5}\Omega_2^{1.5} \left( 1+3 \left(\beta-\tau\right)^2 \right. \nonumber \\
& \quad \left. +\beta\left(\beta-\tau\right)^3+3\beta\left(\beta-\tau\right) \right) +\Omega_2^2 \left( 1+6(\beta-\tau)^2 \right. \nonumber \\
& \quad \left. +(\beta-\tau)^4 \right) \left( \frac{m+1}{m} \right).
\end{align}
Here, we use the $n$-th order moment of $\alpha_i$, i.e. $\displaystyle\mathbb{E}\{\alpha_i^n\}=\frac{2m^m}{\Gamma(m)\Omega_i^m}\int_0^{\infty}\!\! z^{2m+n-1}e^{-\frac{mz^2}{\Omega_i}} dz \! = \!\frac{\Gamma\left(m+\frac{n}{2}\right)}{\Gamma(m)}\left( \frac{\Omega_i}{m} \right)^{\frac{n}{2}}$, which can be obtained in a way similar to \eqref{val1}.

Finally, by substituting \eqref{pow2} and \eqref{pow4} in \eqref{redef1}, we obtain $P_{\rm harv}$ as stated in \eqref{theq1}.

\section*{Acknowledgment}

This work was co-funded by the European Regional Development Fund and the Republic of Cyprus through the Research and Innovation Foundation, under the project INFRASTRUCTURES/1216/0017 (IRIDA). It has also received funding from the European Research Council (ERC) under the European Union’s Horizon 2020 research and innovation programme (Grant agreement No. 819819).

\bibliographystyle{IEEEtran}
\bibliography{refs}

\end{document}